\documentclass[a4paper,11pt]{article}

\usepackage[utf8]{inputenc}
\usepackage[british]{babel}
\usepackage[margin=2.5cm]{geometry}
\usepackage[shortlabels]{enumitem}
\usepackage{amsmath, amsthm, amssymb, amsfonts}
\usepackage{mathtools}
\usepackage{thmtools, thm-restate}
\usepackage{hyperref}
\usepackage{stmaryrd}
\usepackage{bbm}
\usepackage{tikz}
\usepackage{verbatim}

\usepackage{subcaption}

\usepackage{mdframed}

\usepackage[capitalize, nameinlink]{cleveref}

\usepackage[textsize=tiny]{todonotes}

\newcommand{\abs}[1]{\ensuremath{\left| #1 \right|}}

\declaretheorem[parent=section, name=Theorem]{thm}
\declaretheorem[sibling=thm]{lemma}

\declaretheorem[sibling=thm]{proposition}
\declaretheorem[sibling=thm]{claim}

\declaretheorem[sibling=thm]{remark}
\declaretheorem[sibling=thm,style=definition]{definition}

\newtheorem*{remark*}{Remark}

\definecolor{RoyalAzure}{rgb}{0.0, 0.22, 0.66}
\definecolor{ForestGreen}{rgb}{0.13, 0.55, 0.13}

\renewcommand{\leq}{\leqslant}
\renewcommand{\geq}{\geqslant}

\newcommand{\eps}{\varepsilon}

\DeclareMathOperator*{\E}{\mathbb{E}}
\renewcommand{\Pr}{\mathbb{P}}

\hypersetup{
  colorlinks,
  linkcolor={red!60!black},
  citecolor={green!50!black},
  urlcolor={blue!80!black}
}

\title{Navigating an Infinite Space with Unreliable Movements}

\author{Anders Martinsson\\
		ETH Z\"{u}rich\\
		anders.martinsson@inf.ethz.ch
		\and
		Jara Uitto\thanks{Supported by ERC Grant No. 336495 (ACDC)}\\
		ETH Z\"{u}rich and University of Freiburg\\
		jara.uitto@inf.ethz.ch		
}
\date{}

\begin{document}
\maketitle

\begin{abstract}
We consider a search problem on a $2$-dimensional infinite grid with a single mobile agent. The goal of the agent is to find her way home, which is located in a grid cell chosen by an adversary. Initially, the agent is provided with an infinite sequence of \emph{instructions}, that dictate the movements performed by the agent. Each instruction corresponds to a movement to an adjacent grid cell and the set of instructions can be a function of the initial locations of the agent and home. The challenge of our problem stems from faults in the movements made by the agent. In every step, with some constant probability $0 \leq p \leq 1$, the agent performs a random movement instead of following the current instruction. 
%	
%Our main interest lies in the influence that the parameter $p$ has on the hitting time of home. In particular, we are investigate whether there exists a set of instructions that admits a finite expected hitting time. Clearly, if $p = 0$, the problem is trivial; the instructions can simply contain the shortest path home. If $p = 1$, the movements correspond to a random walk. In this case, folklore tells us that the agent will eventually find home but the time to do so has a very slowly decaying upper tail of $\Pr(t > T) = \Theta(1/\log t))$, where $T$ is a random variable that captures the time when home is found. In particular, expected time to find home tends to infinity.

%This paper provides two results on this problem. First, we show that for some values of $p$, there does not exist any set of instructions that guide the agent home in finite expected time. Second, we complement this impossibility result with an algorithm that, for sufficiently small values of $p$, yields a finite expected hitting time for home.
%In particular, we show that for any $p < 1$, our approach gives a hitting rate that decays polynomially (instead of logarithmically) as a function of time. In that sense, our approach is far superior to a standard random walk in terms of hitting time.
%The main contribution and take-home message of this paper is to show that, for some value of $0.01139\dots < p < 0.6554\ldots$, there exists a phase transition on the solvability of the problem.

%%removed logarithmically
This paper provides two results on this problem. First, we show that for some values of $p$, there does not exist any set of instructions that guide the agent home in finite expected time. Second, we complement this impossibility result with an algorithm that, for sufficiently small values of $p$, yields a finite expected hitting time for home.
In particular, we show that for any $p < 1$, our approach gives a hitting rate that decays polynomially as a function of time. In that sense, our approach is far superior to a standard random walk in terms of hitting time.
The main contribution and take-home message of this paper is to show that, for some value of $0.01139\dots < p < 0.6554\ldots$, there exists a phase transition on the solvability of the problem.
\end{abstract}

\newpage

\section{Introduction}
We study a search problem on an infinite $2$-dimensional grid, where the task of a mobile agent is to find \emph{home}, i.e., a designated grid cell chosen by an adversary.
The agent is endowed with a sense of orientation, i.e., at all times, the agent is able to distinguish between the four globally consistent cardinal directions.
The time is divided into discrete time steps and in every step, the agent is able to move to an adjacent grid cell.
Initially, the agent is placed in the origin of the grid and the agent operates according to an infinite sequence of \emph{instructions}, where an instruction corresponds to a movement to one of the cardinal directions.
The set of instructions can be a function of the location of home.
We incorporate errors in navigation during the search by introducing a parameter $0 \leq p \leq 1$, that corresponds to the probability of making a mistake in any time step $t$.
More precisely, in each time step, with probability $(1 - p)$ the agent executes the next instruction in its sequence and, with probability $p$, ignores the instruction and instead moves to an adjacent cell chosen uniformly at random.

The search process can be seen as a random walk augmented by a set of deterministic moves.
The case of $p = 0$ is rather unexciting, since the shortest path home yields deterministic strategy that clearly has a finite expected hitting time.
Also, if we set $p = 1$, we get into a setting where the agent simply follows an unbiased random walk.
This case is also degenerate in the sense that it is well-known that an unbiased random walk has an infinite expected hitting time on \emph{any} cell of the grid.
Hence, it is natural to ask whether we can bias the random walk with some determinism in order to obtain better (i.e., smaller) hitting times. 
Indeed, our main research questions is: Which values of $p$ admit a finite expected hitting time for home?

\begin{definition}
	Let a sequence of \emph{instructions} consisting of an infinite walk $\{x_t\}_{t=0}^\infty$ in $\mathbb{Z}^2$ with $x_0=0$ and a fixed probability $p$ that the agent makes a mistake be given. We call the random process defined by our search protocol a \emph{guided random walk}. Formally, this is defined as a countable state Markov chain $\{X_t\}_{t=0}^\infty$, where $X_0 = 0$ and
\begin{equation*}
X_{t+1}-X_t = \begin{cases} x_{t+1}-x_t&\text{ with probability }(1-p)\\
\operatorname{Unif}\left( (\pm 1, 0), (0, \pm 1)\right)&\text{ with probability }p \ .
\end{cases}\end{equation*}
\end{definition}

Our search problem is inspired by search tasks with limited information and traveling in unmarked terrains.
In practical settings, there are inherently errors in distance estimations and orientations.
In particular, in natural settings, such as foraging behavior of ants, the agents need to operate on imperfect information and need to be able to tolerate errors.
For example, a foraging desert ant needs to find its nest without knowing its exact location and without being able to leave markers on the ground.
However, when heading back to the nest from their search, the ants are able to estimate their distance to the nest quite well~\cite{Wittlinger1965} and have a good sense of the direction towards the nest~\cite{Wehner1981}.

\paragraph{Observations and Some Notation.} 
An initial observation about the guided random walk is that, for any sequence of instructions, the average position of the walk after $t$ steps is $(1-p)x_t$ with a standard deviation of $\Theta(\sqrt{pt})$.
We note that finding home becomes strictly harder when $p$ increases and similarly, easier when $p$ gets smaller in the following sense:
If home cannot be found for some value of $p_0$, the same holds for any $p \geq p_0$ and conversely, a strategy that works for $p_0$ can be adapted to work for any $p \leq p_0$.
For the remainder of the paper, we denote the location of home by $x_{\textrm{home}}$ and the random variable $T:=\inf\{t : X_t=x_{home}\}$ captures the time agent finds home.

\paragraph{Contributions and Technical Challenges.}
In this paper, we examine the proposed navigation problem from two different angles. Our first contribution, captured in \Cref{thm: impossibru}, is to show that no sequence of instructions guarantees finite expected hitting time if $p$ is sufficiently close to one.
The technical challenge is to show that the position of the guided random walk after $t$ steps is anti-concentrated in the sense that the probability of $X_t$ to equal a certain vertex in $\mathbb{Z}^2$ cannot be higher than $\Theta(1/pt)$. Moreover, this picture does not change too much if one conditions on certain properties of the history of $X_t$.

Guided by this analysis, we derive an upper bound on the probability that, given that a guided random walk has not found home at time $t$, it finds home before time $(1+\epsilon)t$. Informally speaking, this needs to be at least $\epsilon$ on average in order for the hitting time to be finite. We get the following result.

\begin{thm}
	\label{thm: impossibru}
	Let $T$ be the random variable that captures the time step where the agent finds home. If $p > 0.7805\dots$, then $\E[T] = \infty$.
\end{thm}
\begin{proof}
	Using the estimate from Proposition \ref{prop:R4estimate} we see that the analytic condition from  Proposition \ref{prop:impossibilityanal} fails for $p> 0.7805\dots$. Hence the expected hitting time of any guided random walk must be infinite in this range.	
\end{proof}
\begin{remark*}
In fact, by using computer assistance to calculate the bound on $R_\tau(p)$ from Proposition \ref{prop:R4estimate} for higher values of $\tau$, it is possible to improve this bound to $p>0.6554\dots$.
\end{remark*}

%Without loss of generality, we can think of the agent moving deterministically and assume that the location of home behaves like a random walk.\jatodo{is this true?}

Our second contribution, captured in \Cref{thm: possibru}, is an algorithm that finds a sequence of instructions such that, for sufficiently small $p$, the agent is able to locate home in expected finite time. The intuition behind our strategy is relatively simple: After $t$ steps of a guided random walk, we expect the agent to be at distance $O(\sqrt{pt})$ from its estimated position. Keeping this in mind, we repeatedly tell the agent to perform sweeps covering an area within this distance $\Theta(\sqrt{pt})$ of the current estimation of the location of home. Once the sweep is done, we re-estimate the location and continue to the next sweep.

%Standard results on random walks tell us that after $t$ steps of a random walk, with high probability, we have an $O(\sqrt{t})$ drift from the original location. Hence, our estimate of the current location of home has an error of $O(\sqrt{t})$.

There are two main challenges of this analysis. We first show that, if the true position of the agent at the beginning of a sweep is close to its estimated position, then the sweep has a good chance to get the agent home at some time step. Second, we show that $X_t$ is unlikely to deviate too much from its estimated position, even when conditioning on the agent not having been home in any previous phase. This makes use of a Chernoff-type bound, which is updated recursively to account for dependencies.

%The two main challenge in the analysis is to first show that, if the true position of the agent

%is to deal with dependencies.
%Informally, whenever we discover that home is not in the current grid cell that the agent occupies, the probability distribution according to which the home is located is biased away from the current location of the agent.
%Hence, it could be that there are time steps during the execution that are very bad in the sense that the probability of finding home is very small.
%To mend this, we split the execution into many phases and show that there is ``enough randomness'' between the phases to consider them (almost) independent.

\begin{thm}
	\label{thm: possibru}
	Let $T$ be the random variable that captures the time step when the agent finds home. There exists a sequence of instructions such that, for any $p < 0.01139\dots$ and any $x_{home}$, the guided random walk following these instructions satisfies $\E[T] < \infty$.
\end{thm}%\jatodo{Add proof here and change to a sharper number if possible DONE}
\begin{proof}
By Proposition \ref{prop:possibrucondition}, the strategy described in Section \ref{sec: possibru} with parameter $a>0$ yields a finite expected hitting time for any $p<p_0$ if
$$\frac{1}{2a^2} \left( 1-4e^{-\frac{a^2}4+2}\right) > \frac{p_0 \sqrt{3-2p_0}}{(1-p_0)^2}.$$
Optimizing the left-hand side over $a$, we get a maximum of $0.02011\dots$ at $a=4.566\dots$. Using this value of $a$, the analytic condition holds for any $p_0< 0.01139\dots$, as desired.
\end{proof}

\begin{remark}
By a slightly more careful consideration of the analysis of our strategy, one can see that, for any $p$, the asymptotic hitting probability converges significantly faster than in the unbiased case. More precisely, it is well-known that an unbiased random walk satisfies $\mathbb{P}(T>t) = \Theta(1/\log t)$ whereas, for any $p < 1$, our strategy obtains a hitting rate of $\mathbb{P}(T>t) = O(t^{-\alpha})$ for $\alpha=\Theta\left( (1-p)/\sqrt{p} \right)$.
\end{remark}

\vspace{2pt}
\begin{mdframed}[backgroundcolor=gray!20,topline=false,
  rightline=false,
  leftline=false,bottomline=false] 
\textbf{A Phase Transition Behavior:}

\noindent Let $\E_p[T]$ denote the expected time to find home given an error parameter $p$. Then, there is a constant value $0 < p^* < 1$ such that
	\begin{enumerate}
		\item for any $p < p^*$, there exists a sequence of instructions that yields $\E_p[T] < \infty$ and
		\item for any $p > p^*$, all sequences of instructions yield $\E_p[T] = \infty$.
	\end{enumerate}
\noindent Our work shows the existence of such a threshold and narrows down the range of possible values of the threshold. Finding the exact value of $p^*$ is left as an intriguing open question.
\end{mdframed}
\vspace{2pt}   

\subsection{Related Work}
Our work falls under a much wider umbrella of graph exploration.
In the standard graph exploration setting a mobile agent (or a group of them) is placed on one of the nodes of the graph.
The goal of the agent is to \emph{explore} the graph, i.e., visit every node or every edge of the graph by traversing the graph along its edges.
There are many variants of graph exploration.
For example, in directed~\cite{Deng1999, Albers2000} and undirected graph exploration~\cite{Duncan2006, Disser2016} the edge traversals are uni- and bidirectional, respectively.
Another distinction is to consider graphs where nodes are equipped with unique identifiers~\cite{Panaite1998} and anonymous graphs~\cite{Rollik1979, Bender1994}.
Our setting corresponds to the undirected and anonymous setting.
Typically, the efficiency of a graph exploration algorithm is measured in terms of its \emph{time complexity}, i.e., how many edge traversals are required to complete the exploration task.

Another measure of interest is the number of bits of memory the agent has~\cite{Fraigniaud2004}.
For example, Fraigniaud et al.~\cite{FraigniaudIPPP2005} showed that a deterministic agent needs $\Theta(D \log \Delta)$ memory to explore a graph with diameter $D$ and maximum degree $\Delta$.
A measure close to the number of bits of memory is the precision of randomness the agent is allowed to use~\cite{Lenzen2014}.
In the case of a random agent in a finite graph, the agent can simulate a random walk without any memory (and with $O(\log \Delta)$ precision) and achieve a polynomial hitting time for every node~\cite{Aleliunas1979}.
In a $2$-dimensional infinite grid, it is known that the random walk is null-recurrent, i.e., it will eventually reach every node, but the hitting time tends to infinity.
If the dimension grows larger, the walk is not guaranteed to reach every node.
In our setting, the agent has no memory, but blindly follows the instructions that it is given.
Cast in the terms of random walks, our research question is whether we can improve the random walk to have a finite expected hitting time on some node in the grid.

A classic graph exploration problem is the \emph{cow-path} problem proposed by Baeza-Yates, Culberson, and Rawlins~\cite{Baeza-Yates1993}.
The cow is located in some node of a path and its task is to find a dedicated node on the path.
In the case of the cow-path problem, the cow needs to fix her strategy without knowing the location of the dedicated node.
The authors showed that a simple spiral search strategy achieves a constant competitive ratio.
In a sense, our approach is an adaptation of the spiral search strategy that is resilient to the challenges caused by the random steps.

Also, exploring an infinite grid with a group of agents has been intensively studied~\cite{Lopez-Ortiz2001, Rollik1979, Fraigniaud2006, Emek15}.
A recent work that falls relative close to ours is by Cohen et al.~\cite{Cohen2017}, where they showed that even allowing two agents and a constant amount of memory per agent (i.e., essentially two correlated random walks) cannot make the expected hitting time finite.

\paragraph{Paper's Organization.}
In \Cref{sec: impossibru}, we provide the necessary tools for the proof of \Cref{thm: impossibru}.
Then, \Cref{sec: possibru} gives the means to prove \Cref{thm: possibru}.

\section{The Impossibility Result}\label{sec: impossibru}
In this section, we give the necessary claims to prove \Cref{thm: impossibru}.
We start with a technical lemma, where we derive an anti-concentration property of a guided random walk using Fourier analysis.
Then, the second step is to introduce a measure, $R_\tau(p)$, that very roughly speaking, bounds from below the drift of the guided random walk within $\tau$ steps, where $\tau$ will be chosen to be a small constant.
A more formal statement will be provided later in this section.
Then, as the third step, we split the execution into exponentially growing intervals.
We show that, conditioned on the event that home was not found prior to the current interval, home is unlikely to be found in the current interval.
%The last step is to show that if the history of the execution were to have a big influence on the probability of finding home, then it must be the case that the measure $R_\tau(p)$ is small yielding a contradiction. This is complete nonsense, sorry...

% we show that, if $p$ is sufficiently close to $1$, then there is no sequence of instructions such that yield a finite expected hitting time. We shall here show this for $p\geq 0.83$, but as will be elaborated on below, the same argument can be extended down to $p\geq 0.655\dots$

\subsection{Step 1: An Anti-Concentration Result for Guided Random Walks} 

\begin{proposition}\label{prop:fourier}
Let $X_t$ be a denote a guided random walk with any sequence of instructions and any (possibly random) initial position. Then
\[
	\max_{x \in\mathbb{Z}^2} \mathbb{P}( X_t = x) \leq \frac{2}{\pi t p\sqrt{3-2p} }+ O\left( \frac{(\ln t)^4}{t^3} \right) \ .
\]
\end{proposition}
\begin{proof}
In what follows, we consider $t>0$ fixed, and let $X_0=0$. For any $\xi, \eta\in \mathbb{R}$, let
\[
	f(\xi, \eta) := \mathbb{E} \left[ e^{i (X_t^1 \xi + X_t^2 \eta)} \right] \ .
\]
This function is $2\pi$-periodic in both $\xi$ and $\eta$ and, for any $x=(x_1, x_2)$, satisfies
\begin{equation}\label{eq:fourierbound}
\mathbb{P}( X_t=x ) = \frac{1}{(2\pi)^2} \int_0^{2\pi}\int_0^{2\pi} f(\xi, \eta) e^{-i(x_1\xi + x_2\eta)}\,d\xi\,d\eta \leq \frac{1}{(2\pi)^2} \int_0^{2\pi}\int_0^{2\pi} \abs{f(\xi, \eta)}\,d\xi\,d\eta \ .
\end{equation}
Moreover, we can rewrite this function as
$$f(\xi, \eta) = f_N(\xi, \eta)^{t_N}f_S(\xi, \eta)^{t_S}f_W(\xi, \eta)^{t_W}f_E(\xi, \eta)^{t_E} \ ,$$
where $t_N, t_S, t_W, t_E$ are the numbers of times up until $t$ when the agent has been instructed to walk north, south, west and east respectively, and
\begin{equation}\label{eq:fnswe}\begin{split}
f_N(\xi, \eta) &= \left( 1-\frac{3p}{4} \right) e^{i\eta} + \frac{p}{4} e^{-i\eta} + \frac{p}{4} e^{-i\xi} + \frac{p}{4} e^{i\xi} \ , \\
f_S(\xi, \eta) &= \frac{p}{4} e^{i\eta} + \left( 1-\frac{3p}{4} \right) e^{-i\eta} + \frac{p}{4} e^{-i\xi} + \frac{p}{4} e^{i\xi} \ , \\
f_W(\xi, \eta) &= \frac{p}{4} e^{i\eta} + \frac{p}{4} e^{-i\eta} + \left( 1-\frac{3p}{4} \right) e^{-i\xi} + \frac{p}{4} e^{i\xi} \ , \\
f_E(\xi, \eta) &= \frac{p}{4} e^{i\eta} + \frac{p}{4} e^{-i\eta} + \frac{p}{4} e^{-i\xi} + \left( 1-\frac{3p}{4} \right) e^{i\xi} \ . \\
\end{split}\end{equation}

One can observe that the functions in \eqref{eq:fnswe} have absolute value one if and only if, up to periodicity, $(\xi, \eta)=(0, 0)$ or $(\pi, \pi)$, and have absolute value less than one everywhere else. Moreover, by Taylor expanding around either $(\xi_0, \eta_0)=(0, 0)$ or $(\xi_0, \eta_0)=(\pi, pi)$ we get the estimates $$\abs{f_N(\xi_0+\xi, \eta_0+\eta)}, \abs{f_S(\xi_0+\xi, \eta_0+\eta)} = e^{-\frac{1}{2}( p(\frac{3}{2}-p) \eta^2 + \frac{p}{2}\xi^2) + O(\xi^4+\eta^4)}$$ and $$\abs{f_W(\xi_0+\xi, \eta_0+\eta)}, \abs{f_E(\xi_0+\xi, \eta_0+\eta)} = e^{-\frac{1}{2}( p(\frac{3}{2}-p) \xi^2 + \frac{p}{2}\eta^2) + O(\xi^4+\eta^4)}.$$

Let $\epsilon>0$ be any sufficiently small constant. We can bound the integrand in the right-hand side of \eqref{eq:fourierbound} by the above expressions for any point $\epsilon$-close to $(0, 0)$ or $(\pi, \pi)$, and by $e^{-\Omega(\epsilon^2 t)}$ everywhere else to obtain
\begin{align*}
&\int_0^{2\pi}\int_0^{2\pi} \abs{f(\xi, \eta)}\,d\xi\,d\eta \\
&\leq 4\pi^2 e^{-\Omega(\epsilon^2 t)} + 2e^{O(\varepsilon^4)} \int_{-\epsilon}^\epsilon\int_{-\epsilon}^\epsilon e^{\frac{1}{2} C_H\xi^2} e^{\frac{1}{2} C_V \eta^2} \,d\xi\,d\eta\\
&\qquad\leq 4\pi^2 e^{-\Omega(\epsilon^2 t)} + \frac{4 (1+O(\varepsilon^4))}{\sqrt{C_HC_V}} \int_{-\infty}^\infty \int_{-\infty}^\infty e^{-\xi^2-\eta^2} \\
&\qquad= \frac{4\pi}{\sqrt{C_H C_v}}(1+O(\epsilon^4)) + 4\pi^2 e^{-\Omega(\epsilon^2 t)},
\end{align*}
where $C_H=(t_E+t_W) \frac{p}2(3-2p) + (t_N+t_S) \frac{p}{2}$ and $C_V=(t_E+t_W) \frac{p}{2}(3-2p) + (t_N+t_S) \frac{p}{2}$. The product $C_HC_V$ is minimized when either $t_N+t_S=t$ and $t_E+t_W=0$ or $t_N+t_S=0$ and $t_E+t_W=t$, which implies that $\sqrt{C_HC_V} \geq t \frac{p}{2} \sqrt{3 - 2p}$. Letting $\epsilon$ equal a sufficiently large constant times $\sqrt{\frac{\ln t}{t}}$ yields the desired bound.
\end{proof}

\subsection{Step 2: Bounding Hit Probabilities} In order to make use of the estimate from Step 1, we would like to be able to bound the probability that $X_t=x$ for some $t \in [t_1, t_2]$ in terms of the expected number of such $t$. One simple way of obtaining such a bound would be Markov's inequality. However, the bound one would obtain in this way is slightly too weak for our purposes. So, in this step we will derive a stronger bound. This makes use of the fact that, due to the nature of a (mostly) random walk, if the agent reaches $x_{\mathrm{home}}$, it is likely to return there in subsequent time steps.

\begin{proposition}\label{prop:hitreturnbound}
For any positive integer $\tau$, let $$R_\tau(p) := \min \mathbb{E}[\#\{t \in [\tau] : X_i=0\} \mid X_0=0]$$
be the expected number of times between $0$ and $\tau$ a guided random walk will return to to its initial position, where the minimum is taken over all possible sequences of instructions the agent could receive.

Then, for any $0 \leq t_1 \leq t_2$, and any $x\in \mathbb{Z}^2$, a guided random walk with any (possibly random) initial position $X_0$ satisfies
\begin{equation}\label{eq:hitreturnbound}
\mathbb{P}( X_t=x \text{ for some } t\in [t_1, t_2]) \leq \mathbb{E}[\#\{t \in [t_1, t_2+\tau]: X_t=x\}]/R_\tau(p) \ .
\end{equation}
\end{proposition}
\begin{proof}
Conditioned on the event that $t'$ is the smallest integer $t\in [t_1, t_2]$ such that $X_t=x$, the process $Y_s:=X_{t'+s}-X_{t'}$ behaves as a guided random walk starting in the origin. Thus the expected number of times $s\in[0, \tau]$, that is times $t\in [t', t'+\tau]$ when $X_t=X_{t'}=x$ is at least $R_\tau(p)$.
\end{proof}

\paragraph{Estimating $R_\tau(p)$.}
Next, we compute estimates for $R_{\tau}(p)$. We derive a general upper bound for $R_\tau(p)$ for any $\tau$, and then compute this quantity explicitly for $\tau=4$, which is what will be used later in the proof. Using larger values of $\tau$ would result in in stronger bounds on $p$, but as the coefficients get increasingly hard to compute by hand, we will not do this here.

%First, we give a more detailed proof for $\tau = 4$ and $p = 5/6$, since this claim can be verified by using a pen and paper.
%Then, we give a sharper estimate that we obtained through computer assistance.
%The second claim is not really feasible to verify by hand and hence, to convince the reader about the feasibility of our approach, we decided to add the weaker claim and its proof.

\begin{proposition}\label{prop:R4estimate} 
Let $W_r(a, b)$ denote the number of $r$-step walks in the grid from $0$ to $(a, b)$ and let $W_{r,s}$ denote the minimum of $W_r(a, b)$ over all $a, b$ such that $\abs{a}+\abs{b}\leq s$ and $a+b\equiv s \text{ mod }2$. Then, for any even $\tau>0$,
\begin{equation}
R_{\tau}(p) \geq \sum_{k=0}^{\tau/2} \sum_{r+s=2k} {r+s\choose r} (1-p)^s \left(\frac{p}{4}\right)^r W_{r,s} \ .
\end{equation}
In particular,
\[
	R_4(p) = 1 + \frac{1}{4}p^2+\frac{1}{2}p(1-p) + \frac{9}{64}p^{4}+ \frac{9}{16}p^{3}(1-p)+\frac{3}{8} p^2(1-p)^2 \ .
\]
%$R_4(5/6)>1.33$. %$E_{30}(2/3) > 1.5$.
\end{proposition}
\begin{proof} Let us initially consider any fixed $\tau$, $p$, and any sequence of instructions $x_0, x_1, \dots$. By linearity of expectation, the expected number of times $X_t=0$ for $t\in [\tau]$ is given by $\sum_{k=0}^{\tau/2} \mathbb{P}(X_{2k}=0).$

Condition on the event that the agent has taken precisely $r$ random steps before time $2k$, and the times at which these occurred. Here we consider a step as random if it is produced by walking in a random direction \emph{even if} this direction turns out to match the instructions.
Given this information, we can split the movement of the agent up until time $2k$ into two terms: the translation produced by the $s:=2k-r$ non-random steps, which is some vector $(-a, -b)$ where $\abs{a}+\abs{b}\leq s$ and $a+b\equiv s\text{ mod }2$, and the translation produced by the $r$ random steps, which would then have to equal $(a, b)$ in order for $X_{2k}=0$.

Using this, we get the lower bound
\begin{equation}\label{eq:expreturnformula}
\mathbb{P}(X_{2k}=0) \geq \sum_{r+s=2k} {r+s\choose r} (1-p)^s \left(\frac{p}{4}\right)^r W_{r,s}.
\end{equation}
The first formula in the statement of the proposition immediately follows.

%It seems intuitively clear that $W_r(a, b)$, with restrictions as above, is always minimized by $(\pm s, 0), (0, \pm s)$, which would mean that the optimal instructions to minimize the expected number of returns to the starting vertex is to always walk in the same cardinal direction. However, as explicitly finding $W_{r,s}$ with computer assistance is not significantly harder than finding $W_r(s, 0)$, we omit a proof of this.

It is easily seen that $W_{0,0}=1$, $W_{2,0}=4, W_{1,1}=1,$ and $W_{0,2}=0$, and it can further be seen without too much effort that $W_{4,0}=36, W_{3,1}=9, W_{2,2}=1,$ and $W_{1, 3}=W_{0,4}=0$, which yields the second formula in the statement of the proposition. Equality follows from the fact that the lower bound is attained exactly if the instructions follow a straight line in any cardinal direction.

% Extending this analysis with computer assitance for times $0, 2, \dots, 30$ yields the desired expression.\footnote{In fact, this estimate would be enough to prove that finite expected hitting time is impossible for $p>0.8241$. Just plug this estimate into the analytic condition in Proposition 3.8 to verify.}
\end{proof}
\begin{comment}
\begin{claim}\label{claim: R30estimate}
	$R_{30}(2/3) > 1.5$
\end{claim}
\begin{proof}[Proof Sketch]
	To improve on the estimate in \Cref{claim:R4estimate}, we implemented a simple script that, using dynamic programming, computes an estimate of $R_{30}(2/3)$.
	Otherwise, the proof follows the same steps as the proof for $R_4(5/6)$.
\end{proof}
\begin{remark*}
	By using more computing power, it is possible to improve on the estimates we provide. 
	However, this approach is unlikely to provide us with the exact bounds and hence, new insights are required.
\end{remark*}
\end{comment}

\subsection{Step 3: Handling Dependencies}

So far in the proof, we have only dealt with properties of the guided random walk at some time $t$ \emph{without} any conditions on the history of the process up to time $t$. Here, we combine the results from the previous two steps to obtain an estimate for the probability that $T<(1+\eps)t$ given that $T\geq t$.

To achieve this, we pick a time $s \in [0, t]$. (It turns out optimal to choose $s\approx \frac{t}{2}$.) Intuitively, if the agent has not found home yet at time $s$, and we know that it is unlikely for her to get home during the interval $[s, t]$, then conditioning on her not finding home during this interval should not affect the behavior of the agent after time $s$ too much.
Let $x_{\mathrm{home}}$ denote the location of home.

\begin{proposition}\label{prop:hitconduncond}For any $\epsilon>0$ and any $0\leq s \leq t$, we have
$$\mathbb{P}( T < (1+\epsilon)t \mid T \geq t) \leq \frac{\mathbb{P}(X_r = x_{\mathrm{home}}\text{ for some }r\in [t, (1+\epsilon)t) \mid T\geq s )}{\mathbb{P}(T\geq t \mid T\geq s)} \ .$$
\end{proposition}
\begin{proof}
Let $A, B$ and $C$ denote the events that $X_r=x_{\mathrm{home}}$ for some $r\in [t, (1+\epsilon)t)$, some $r \in [s, t)$, and some $r < s$ respectively. Then, we can rewrite the inequality above as
$$\mathbb{P}( A \mid \bar{B} \cap \bar{C} ) \cdot \mathbb{P}(\bar{B} \mid \bar{C}) \leq \mathbb{P}(A \mid \bar{C}) \ .$$
But by the multiplication rule for conditional probabilities, the left-hand side equals $\mathbb{P}(A\cap \bar{B} \cap \bar{C}) / \mathbb{P}(\bar{C}) = \mathbb{P}(A\cap \bar{B} \mid \bar{C})$, which is clearly less than the right-hand side.
\end{proof}

By combining Propositions \ref{prop:fourier} and \ref{prop:hitreturnbound}, we can further obtain an estimate for the numerator in Proposition \ref{prop:hitconduncond}. Here we assume that $t$ is sufficiently large, and $s=\frac{t}{2}+O(\epsilon t)$. Conditioned on $T\geq s$, the agent will move according to an unconditioned guided random walk starting from time $\lfloor s \rfloor$. Hence, for any $r\in [t, (1+\epsilon) t)$, we have
\[
	\mathbb{P}( X_r=x_{\mathrm{home}} \mid T\geq s) \leq \frac{2}{\pi (t-s) p \sqrt{3-2p}} + O((t-s)^{-2}) = \frac{4}{\pi tp\sqrt{3-2p}} + O\left( \frac{\epsilon}{t}+t^{-2} \right) \ .
\]
Moreover, as it is only possible for $X_r=x_{\mathrm{home}}$ at every second time step, due to parity, there are $\frac{1}{2}\epsilon t + O(\tau)$ times during $[t, (1+\epsilon)t+\tau)$ when the hitting probability could be non-zero. Hence,
\[
	\mathbb{E}[\#\{r \in [t, (1+\epsilon)t+\tau) : X_r=x_{\mathrm{home}}\}| T\geq s] \leq  \frac{2\epsilon}{\pi  p \sqrt{3-2p}}+ O\left( \epsilon^2+ \frac{\tau}{t} \right) \ .
\]
Finally, by applying Proposition \ref{prop:hitreturnbound}, we obtain the following:
\begin{proposition}\label{prop:impossibruniceformula} For any $\epsilon>0$, any sufficiently large $t$ and any $s=\frac{t}{2}+O(\epsilon t)$,
\[
 \mathbb{P}( X_r=x_{\mathrm{home}}\text{ for some }r\in [t, (1+\epsilon)t) \mid T\geq s) \leq \frac{2\epsilon}{\pi p \sqrt{3-2p}R_\tau(p)} + O\left( \epsilon^2 + \frac{\tau}t \right) \ .
\]
\end{proposition}

\subsection{Step 4: A Necessary Condition for Finite Expected Hitting Time}
\begin{lemma}\label{lemma: anyhome}
Let $p>0$ be fixed. Suppose there exists a choice of $x_{\mathrm{home}}$ such that $\mathbb{E}[T]=\infty$ for any sequence of instructions. Then, $\mathbb{E}[T]=\infty$ for any choice of $x_{\mathrm{home}}$ apart from the origin.
\end{lemma}
\begin{proof}
Suppose there exists a starting position $x_{bad}\in \mathbb{Z}^2$ such that $\mathbb{E}[T]=\infty$. Then, for any choice of $x_{\mathrm{home}}$ there exists a time $t_{bad}$ such that, with positive probability, $X_{t_{bad}} = x_{\mathrm{home}}-x_{bad}$ and $X_t\neq x_{\mathrm{home}}$ for all $t<t_{bad}$. But the conditional expected hitting time from this state is infinite, thus $\mathbb{E}[T]=\infty$. 
\end{proof}

\begin{proposition}\label{prop:impossibilityanal}
Suppose, for a given $p>0$ and $x_{\mathrm{home}}\neq 0$, there exists a sequence of instructions such that $\mathbb{E}[T] < \infty$. Then the analytic condition
$$1 \leq \frac{4}{\pi p\sqrt{3-2p} R_\tau(p)}$$
holds for any positive integer $\tau$.
\end{proposition}
\begin{proof}
Let $\epsilon>0$ be a sufficiently small constant. By \Cref{lemma: anyhome}, we may, without loss of generality, assume that $x_{\mathrm{home}}$ is sufficiently far from the origin to make the argument below follow through.

For $k=0, 1, \dots$ divide the process up into time intervals $[(1+\epsilon)^k, (1+\epsilon)^{k+1})$. We note that if \begin{equation}\label{eq:catchsmall}
\mathbb{P}( T\geq  (1+\epsilon)^{k+1} \mid T\geq (1+\epsilon)^k) \geq\frac{1}{1+\epsilon}
\end{equation} for all $k$, then $T$ dominates $(1+\epsilon)^{\text{Geom}(1-\frac1{1+\epsilon})}$, which has infinite expectation. Thus, in order for $\mathbb{E}[T]<\infty$ there must exist $k_0$ such that \eqref{eq:catchsmall} is satisfied for all $k=0, 1, \dots k_0-1$, but not for $k=k_0$. By choosing $x_{\mathrm{home}}$ sufficiently far away, we may assume that $k_0$ is much larger than $\epsilon^{-1}$.

Let $t=(1+\epsilon)^{k_0}$ and $s=(1+\epsilon)^l$ where $l$ is the smallest positive integer such that $s\geq \frac{t}{2}$. As $k_0$ corresponds to the first interval where hitting probability is ``high'' we have
$$\mathbb{P}( T \geq t \mid T\geq s) \geq \frac{1}{(1+\epsilon)^{k_0-l} }= \frac{s}{t}\geq \frac{1}{2} \ .$$
Hence, by Propositions \ref{prop:hitconduncond} and \ref{prop:impossibruniceformula}, we have
$$\mathbb{P}(T<(1+\epsilon)t \mid T\geq t) \leq \frac{4\epsilon}{\pi p\sqrt{3-2p}R_\tau(p)} + O\left( \epsilon^2 + \frac{\tau}{t} \right) \ .$$
On the other hand, by assumption of $k_0$, we have
$$\mathbb{P}( T< (1+\epsilon)t \mid T \geq t) > 1-\frac{1}{1+\epsilon} \ .$$
The analytic condition in the statement of the proposition follows by combining these two bounds, and taking the limit as $1\ll \epsilon^{-1} \ll t$.
\end{proof}

% \begin{corollary}
% 	If $p > 0.655\ldots$ \mathbb{E}[T]=\infty$ for any sequence of instructions if $p > 0.655\dots$
% \end{corollary}
% \begin{proof}
% Plug Claim \ref{claim: R30estimate} into Proposition \ref{prop:impossibilityanal} and behold a contradiction.
% \end{proof}

% To wrap up this section, we remark that the bound on $p$ can be directly improved by estimating $R_\tau(p)$ for $\tau>4$ in the same fashion as Claim \ref{claim:R4estimate}. Although tedious to do by hand, extending the analysis to $\tau=30$ implies an infinite expected hitting time whenever $p>0.655\dots$

\section{Finding Home}\label{sec: possibru}%
Throughout this section, we assume that $X_0=0$ and that home is in some vertex at constant distance from the origin.
The goal of this section is to construct a sequence of instructions such that the guided random walk following these instructions has a finite finite expected hitting time to any vertex in $\mathbb{Z}^2$ for any $p < p_0$, where $p_0$ is sufficiently small constant.

Our strategy proceeds in phases, which are performed iteratively until the guided random walk hits the designated vertex in $\mathbb{Z}^2$. 
To simplify the correctness proof, we construct the strategy such that $x_t=0$ at the beginning of each phase. 
Our strategy is an adaptation of the following canonical approach:
At time $t$, fix a box with side length roughly $\sqrt{t}$ around home.
Then, explore every cell in the box and iterate.

\paragraph{Duration of a Phase.}
An issue with the canonical approach is that a phase starting at time $t$ takes roughly $t$ steps to complete.
During the phase, missing home is relatively likely and hence, we end up using a lot of time scanning all of the vertices in the area.
To mend this, in a phase starting at time $t$, we only walk along roughly $\log t$ horizontal lines over a box with side lengths $2A\sqrt{t}$ and $2B\sqrt{t}$ centered around the origin, where $A$ and $B$ are parameters given to the strategy.
By a careful design and choice of these parameters, we can guarantee that the probability of finding home is roughly $(\log t)/\sqrt{t}$.

\subsection{Step 1: The Strategy} \label{sec: strategy}
% The strategy we will use to get grandma home will be divided up into phases. Within each phase we will tell grandma to move in the shape of a thing. This is repeated until eventually grandma is caught.
%
Our strategy is divided into phases and, in each phase, we explore an area (and a few other grid cells) that has a shape of a rectangle.
For simplicity, we will assume that, at the beginning of each phase, we have $x_t=0$. Moreover, we suppose that the agent is likely, say with probability $\Theta(1)$, to be contained in a box $(\pm A \sqrt{t}, \pm B \sqrt{t})$ around the origin.

\subsection{Step 1: The Strategy}
Let $A$ and $B$ be two positive constants to be defined later. For a phase starting at time $t$, we let
\begin{align*}
W&= \left\lceil \frac{1}{1-p_0}A\sqrt{t} \right\rceil,\\
H&= \left\lceil \frac{1}{1-p_0}B\sqrt{t} \right\rceil,\\
G&= \left\lceil \frac{1}{1-p_0}\frac{\sqrt{t}}{\ln t} \right\rceil,\\
N&= \left\lceil \frac{2H}{G} \right\rceil \ .
\end{align*}
We give the following sequence of instructions:
\begin{enumerate}
\item Move $W$ steps west and $H$ steps south to $(-W, -H)$.
\item Move $Z$ steps north where $Z\sim \operatorname{Unif}([G+\lceil t^{1/3}\rceil])$.
\item Form $N$ horizontal lines at horizontal distance $G$ by alternatingly moving $2W$ steps east, $G$ steps north, $2W$ steps west, $G$ steps north, and so on.
\item\label{step (4)} Move $G+\lceil t^{1/3} \rceil-Z$ steps north, and then return to $(0, 0)$.
\end{enumerate}

\begin{remark*}
	In our strategy, Step (4) is performed to make the total number of steps in a phase non-random and hence, the number of steps $Z$ does not bring any extra randomness into our analysis.
\end{remark*}

% A few comments: Step (4) looks like that to ensure that the total number of steps done in a phase is non-random. The most costly part of our approach is trazing all horizontal lines, and therefore the phase takes $(2+o(1))N\cdot W=\frac{(4+o(1)) A\cdot B\cdot \sqrt{t}\ln t}{1-p_0}$ time. Moreover, a $1-o(1)$ fraction goes in west-east direction. This is intentional - reason being that moving in a direction, will increase variance more parallel to the direction than orthogonal to it. Thus keeping the variance as small as possible along one axis (this case north-south) will minimize the area where home is likely to be at.

\subsection{Step 2: Analyzing a Single Iteration}

\begin{proposition}\label{prop:phasegoodifgrandmanear} Fix $p<p_0$. Suppose we start a phase at time $t>0$. Condition on the initial position $X_t$ of the agent at the beginning of the phase being inside the box $(\pm A\sqrt{t}, \pm B\sqrt{t})$. Then with probability at least $\frac{1-o(1)}{G}$, the agent visits home at some point during the phase.
\end{proposition}%\jatodo{Is this true for all $t$ or do we need to assume that $\sqrt{t} > x_{\mathrm{home}}$?}

\paragraph{Notation.}
Before formulating the proof we introduce some notation. Let $U$ be the vector denoting the actual translation that the agent performs during $(1)$, $V_Z$ the translation during $(2)$, and for any $0\leq s \leq (N-1)G+2N\cdot W$ let $W_s$ be the actual translation performed by the agent during the first $s$ steps of $(3)$. Note that $U$, $V_{Z}$ and $\{W_s\}$ are independent.

Our proof will consist of a coupling argument over the different values of $Z$. In doing so, we extend $V_Z$ to a guided random walk  $V_0, V_1, \dots, V_{G+\lceil t^{1/3} \rceil}$ starting in the origin, and where the directions are to always move north. We will show that, with high probability, there is at least one choice of $Z$ that would make the agent get home during the phase.

\begin{lemma}
Let $\{Y_s\}_{s=0}^r$ be any $r$-step guided random walk. Then, with high probability, $\| Y_s - \mathbb{E} [Y_s]\|_\infty \leq 4\sqrt{r}\ln r$ for all $0\leq s \leq r$, where $\|(a, b)\|_\infty := \max(\abs{a}, \abs{b})$.

In particular, with high probability, $U$, $V_0, V_1, \dots$ and $W_0, W_1, \dots$ all deviate from their expectations by at most $o(t^{-1/3})$.
\end{lemma}
\begin{proof}
The first part is a direct consequence of Azuma's inequality, or a suitable version of a standard Chernoff bound. The second statement is a consequence of the first, observing that the random variables can be expressed as three guided random walks with $O(\sqrt{t}\ln t)$ steps, yielding a maximum deviation of $O(t^{1/4}\ln t)=o(t^{1/3})$ from their respective maximum values.
\end{proof}
\begin{proof}[Proof of Proposition \ref{prop:phasegoodifgrandmanear}.]
Let $U, \{V_s\}$ and $\{W_s\}$ be as above, and condition on the event that these never deviate in max-norm from their expectations by more than $o(t^{1/3})$. We show that, under this condition, and assuming $X_t$ is contained in $(\pm A\sqrt{t}, \pm B\sqrt{t})$ there is at least one choice of $Z$ by which the agent would get home at some point during the phase.
For an illustration, refer to \Cref{fig: box}.
\begin{figure}
	\begin{minipage}[b]{0.5\linewidth}
		\centering
		\includegraphics[width=.95\linewidth]{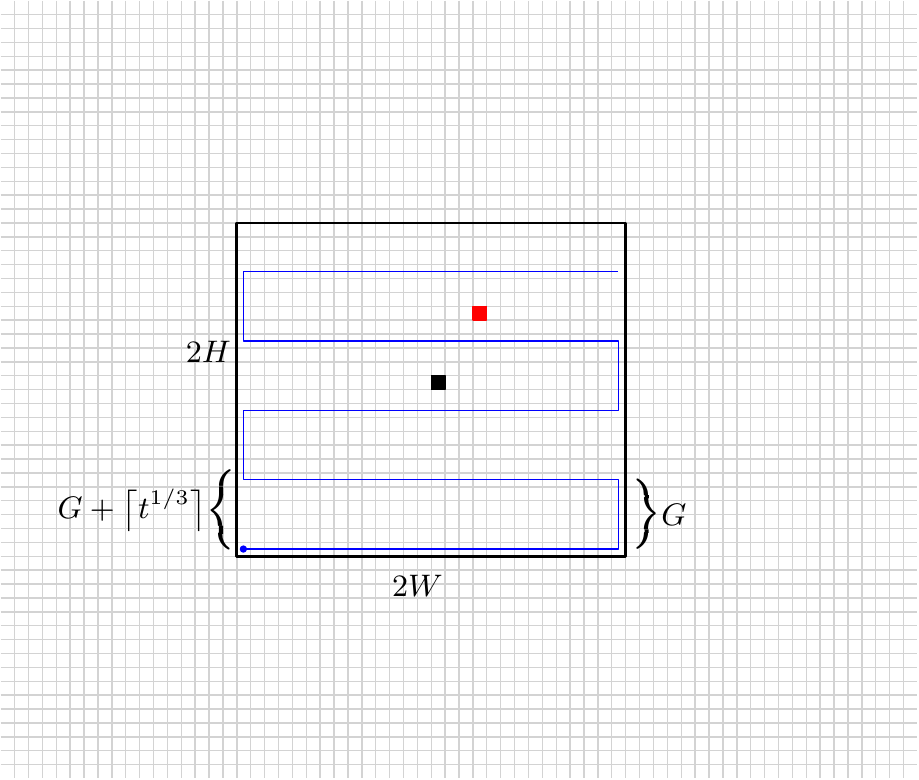}
% 		\subcaption{asd}
	\end{minipage}%
	\begin{minipage}[b]{0.5\linewidth}
		\centering
		\includegraphics[width=.95\linewidth]{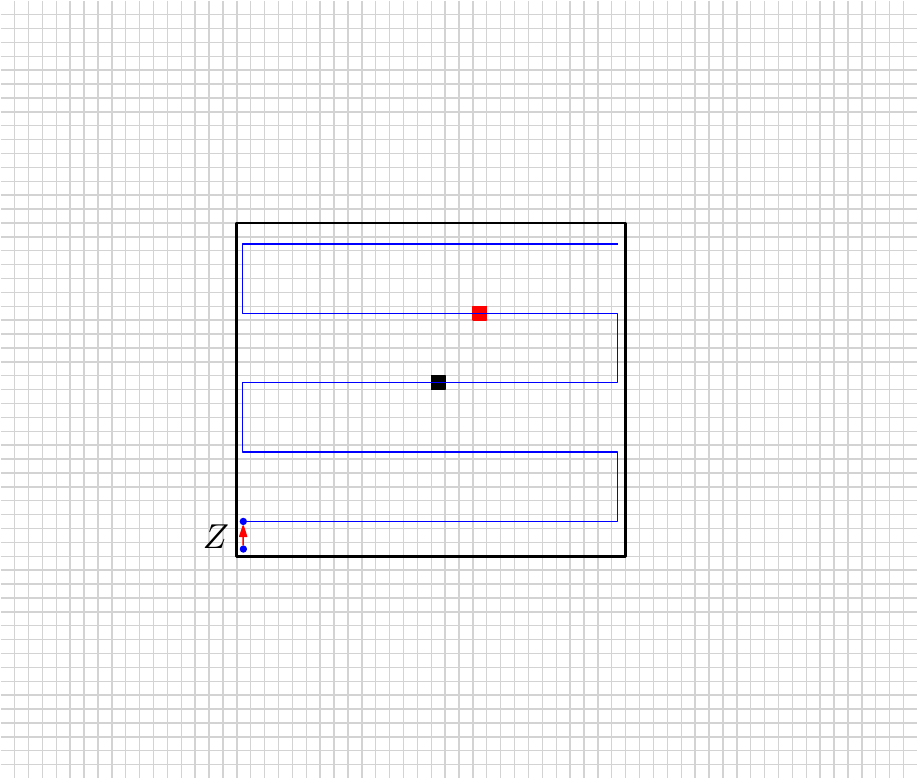}
% 		\subcaption{asd}
	\end{minipage}
	\caption{In the beginning of a phase, starting at time $t_i$, the agent (illustrated by a small blue circle) is told to go to the southwest corner of the box centered at the origin.
The blue zigzag illustrates the $N$ lines of width $2W$. The $Z$ random steps are intentionally left out in the picture on the left. The red square illustrates the point $\frac{1}{1-p}(x_{\textrm{home}}-X_{t_i})$. The significance of this point is that, whenever the instructions go to this point, then the \emph{actual} position of the agent will be close to $x_\textrm{home}$. As this point depends on $X_{t_0}$, it is unknown to the algorithm, but under the assumption that $X_{t_0}$ is not too far away from the origin, this point will be inside the box. Thus by sweeping the box, we are likely to pass close to it at some point.
In the picture on the right, a random choice of $Z$ that leads the agent home is shown by the red arrow. For simplicity, the drift induced by the random steps is left out of the illustrations.}
	\label{fig: box}
\end{figure}

For the claim to hold, it suffices to show that there exist a $z\in [ G+\lceil t^{1/3}\rceil]$ such that
$$U + W_s = x_{\mathrm{home}} - X_t - V_z$$
for some $0\leq s \leq (N-1)G+2N\cdot W$. Consider the set of points the left- and right-hand sides of this expression covers as $s$ and $z$ goes over the respective intervals. In particular, the left-hand side contains $N$ mostly horizontal lines going between $x$-coordinates $$\pm \frac{1-p}{1-p_0} A\sqrt{t} \pm o(t^{1/3})$$ and with $y$-coordinates $$-\frac{1-p}{1-p_0} \left(B\sqrt{t} + i \frac{\sqrt{t}}{\ln t}\right) \pm o(t^{1/3})$$ for each $i=0, 1, \dots, N-1$. In particular, the last line has $y$-coordinate at least $\frac{1-p}{1-p_0}\left( B\sqrt{t} - \frac{\sqrt{t}}{\ln t}\right) \pm o(t^{-1/3}).$

On the other hand, the right-hand side has $x$-coordinates contained in the slightly smaller interval of $\pm A\sqrt{t}\pm o(t^{1/3})$, and with its highest $y$-coordinate contained inside $\pm B\sqrt{t} + o(t^{1/3})$ and its highest and lowest $y$-coordinate differing by $\frac{1-p}{1-p_0} \frac{\sqrt{t}}{\ln t} + (1-p)t^{1/3} + o(t^{1/3})$. As this difference is larger than the gap $\frac{1-p}{1-p_0} \frac{\sqrt{t}}{\ln t}  + o(t^{1/3})$ between two horizontal lines, there must be a $z$ where the two sets intersect, as desired.
\end{proof}

\subsection{Step 3: How Far is the Guided Random Walk from the Origin?}

With \Cref{prop:phasegoodifgrandmanear} at hand, our next goal is to get a handle on how far from the origin we can expect the agent to be at the beginning of a phase, conditioned on home not being found before this phase. 
Fix a parameter $\alpha>0$ and let $t_i$ denote the starting time of phase $i$. To simplify the analysis, we will assume that
\begin{equation}\label{eq:cap}
	\mathbb{P}(T\leq t_{i+1} \mid T>t_i) \leq \alpha \frac{t_{i+1}-t_i}{t_i} \ .
\end{equation}
\paragraph{Model Modification.}
To ensure this bound, we modify the model in the following way.
Every time the agent is about to hit home, there is a small probability that the agent does not ``realize'' that home is found, and simply moves normally in the next step.
Clearly, any such modification would increase $T$, hence any upper bound to $T$ in the modified model would immediately yield the same bound for the original model.

For each cardinal direction, i.e., North, East, South and West, we define
\begin{align*}
m_i^N &= \mathbb{E}\left[ e^{\frac{\lambda}{\sigma_2 \sqrt{t_i}}X^2_{t_i}} \mid T>t_i \right],\\
m_i^E &= \mathbb{E}\left[ e^{\frac{\lambda}{\sigma_1 \sqrt{t_i}}X^1_{t_i}} \mid T>t_i \right],\\
m_i^S &= \mathbb{E}\left[ e^{-\frac{\lambda}{\sigma_2 \sqrt{t_i}}X^2_{t_i}} \mid T>t_i \right],\\
m_i^W &= \mathbb{E}\left[ e^{-\frac{\lambda}{\sigma_1 \sqrt{t_i}}X^1_{t_i}} \mid T>t_i \right],\\
\end{align*}
where $\sigma_1=\sqrt{\frac{p}{2}\left(3-2p\right)}$ and $\sigma_2=\sqrt{\frac{p}{2}}$. That is, $\sigma_1$ is the standard deviation of one step of a guided random walk, parallel to the given direction, and $\sigma_2$ is the standard deviation in the orthogonal direction.

\begin{proposition}\label{prop:chernoffbndgrandma} Let $\alpha>0$ be as above, and let $a>0$ be any positive constant. Conditioned on $T>t_i$, the probability that $X_{t_i}$ is outside the box $(\pm a\sigma_1 \sqrt{t_i}, \pm a\sigma_2 \sqrt{t_i})$ is $4e^{-\frac{1}{4}a^2 + 2\alpha + o(1)}$.
\end{proposition}

Before proving this proposition, we need some ground work.

\begin{claim}\label{claim:genfctrec}
For any $c\in\{N, S, E, W\}$, we have 
$$m_{i+1}^c \leq \left(m_i^c\right)^{\sqrt{\frac{{t_i}}{{t_{i+1}}}}}e^{\left(\frac{1}{2} \lambda^2 + \alpha + o(1)\right)\frac{t_{i+1}-t_i}{t_{i+1}}}$$
\end{claim}
\begin{proof}
We show the inequality for $c=E$. The remaining cases can be shown analogously. We first note that for any $a\in \mathbb{Z}$, we have
\begin{align*}
&\mathbb{P}(X^1_{t_{i+1}} = a \mid T>t_{i+1}) \cdot \mathbb{P}(T>t_{i+1} \mid T>t_{i}) \\&\qquad = \mathbb{P}(X^1_{t_{i+1}} = a\text{ and }T>t_{i+1} \mid T>t_i)\\
&\qquad\leq\mathbb{P} (X^1_{t_{i+1}} = a \mid T>t_i)
\end{align*} 
Thus, $$\mathbb{P}(X^1_{t_{i+1}} = a \mid T>t_{i+1}) \leq \frac{1}{\mathbb{P}(T>t_{i+1} \mid T>t_{i}) } \mathbb{P}(X^1_{t_{i+1}} = a \mid T>t_{i})$$ and therefore $$\mathbb{E} \left[ e^{\frac{\lambda}{\sigma_1 \sqrt{t_{i+1}}}X^1_{t_{i+1}}} \mid T>t_{i+1} \right] \leq \frac{1}{\mathbb{P}(T>t_{i+1} \mid T>t_{i}) }\mathbb{E} \left[ e^{\frac{\lambda}{\sigma_1 \sqrt{t_{i+1}}}X^1_{t_{i+1}}} \mid T>t_{i} \right].$$
Hence, 
\[
	m^E_{i+1} \stackrel{*}{\leq} e^{(\alpha + o(1)) \frac{t_{i+1}-t_i}{t_{i+1}}} \cdot \mathbb{E} \left[ e^{\frac{\lambda}{\sigma_1 \sqrt{t_{i+1}}}X^1_{t_{i}}} \mid T>t_{i} \right] \cdot \mathbb{E} \left[ e^{\frac{\lambda}{\sigma_1 \sqrt{t_{i+1}}}(X^1_{t_{i+1}}-X^1_{t_i})} \right] \ .
\]
Notice that in the estimate marked by $*$, we use the bound given by Inequality~\ref{eq:cap} to get rid of the $1/\mathbb{P}(T > t_{i + 1} \mid T > t_i)$ term. 
By Jensen's inequality
$$\mathbb{E} \left[ e^{\frac{\lambda}{\sigma_1 \sqrt{t_{i+1}}}X^1_{t_{i}}} \mid T>t_{i} \right] \leq \left(\mathbb{E} \left[ e^{\frac{\lambda}{\sigma_1 \sqrt{t_{i}}}X^1_{t_{i}}} \mid T>t_{i} \right]\right)^{\frac{\sqrt{t_{i}}}{\sqrt{t_{i+1}}}} = \left(m_i^E\right)^{\frac{\sqrt{t_{i}}}{\sqrt{t_{i+1}}}}.$$
Moreover, following a standard derivation of a Chernoff bound, we have
$$\mathbb{E} \left[ e^{\frac{\lambda}{\sigma_1 \sqrt{t_{i+1}}}(X^1_{t_{i+1}}-X^1_{t_i})} \right] = \prod_{s=t_{i}}^{t_{i+1}-1} \mathbb{E} \left[ e^{\frac{\lambda}{\sigma_1 \sqrt{t_{i+1}}}\left(X^1_{s+1}-X^1_{s}-(1-p)(x^1_{s+1}-x^1_{s}) \right)} \right]= e^{(\frac{1}{2}+o(1)) \lambda^2 \frac{t_{i+1}-t_i}{t_{i+1}}},$$
where, in the last step, we used that most instructions between $t_{i}$ and $t_{i+1}$ go either west or east, meaning $X^1_{s+1}-X^1_s$ has standard deviation $\sigma_1$ in most steps.
\end{proof}

\begin{claim}\label{claim:boundgenfct}
For any $c\in\{N, S, E, W\}$, we have 
$$m_{i+1}^c \leq e^{\lambda^2+2\alpha + o(1)}.$$
\end{claim}
\begin{proof}
For any fixed $\epsilon>0$, pick $i_0$ sufficiently large such that the exponent in the right-hand side of \Cref{claim:genfctrec} is at most $\left(\frac{1}{2} \lambda^2 + \alpha + \epsilon\right)\frac{t_{i+1}-t_i}{t_{i+1}}$ for any $i\geq i_0$.

Thus, picking a constant $C>1$ such that $m^c_{i_0}\leq C\cdot e^{\lambda^2+2\alpha + 2\epsilon}$ and observing that, as $\sqrt{\cdot}$ is concave, $\sqrt{\frac{{t_i}}{{t_{i+1}}}} \leq 1-\frac{1}{2} \frac{t_{i+1}-t_{i}}{t_{i+1}}$, it follows by induction that
$$ m_{i}^c \leq C^{\sqrt{\frac{t_{i_0}}{t_{i}}}} e^{\lambda^2+2\alpha+2\epsilon}$$
for all $i\geq i_0$, as desired.
\end{proof}

\begin{proof}[Proof of Proposition \ref{prop:chernoffbndgrandma}.]
By \Cref{claim:boundgenfct} and Markov's inequality, we can bound the probability of this event by $4e^{\lambda^2+2\alpha  - \lambda a + o(1)}.$ Minimizing over $\lambda$ yields the claimed probability.
\end{proof}

\subsection{Step 4: Wrapping it up}

Now, we have gathered almost all necessary claims to complete the proof of \Cref{thm: possibru}.
As the last steps, we first provide a (rather) simple observation and then bundle up the statements of this section into \Cref{prop:possibrucondition}, that almost immediately yields the theorem.

\begin{claim}\label{claim:algpowerlaw}
Suppose there exists an $i_0$ such that $\mathbb{P}(T\leq t_{i+1}\mid T>t_i) \geq \alpha \frac{t_{i+1}-t_i}{t_i}$ for all $i\geq i_0$. (That is, combined with Inequality~\ref{eq:cap}, we have equality for all large $i$.) Then $\mathbb{P}(T>t) = O(t^{-\alpha})$.
\end{claim}
\begin{proof}
As $t^{-\alpha}$ is a convex function, we have $\left(\frac{t_{i+1}}{t_i}\right)^{-\alpha} = \left(1+\frac{t_{i+1}-t_i}{t_i}\right)^{-\alpha}\geq 1-\alpha\frac{t_{i+1}-t_i}{t_i} $. This means that $\mathbb{P}(T>t_{i+1} \mid T>t_i) \leq \left(\frac{t_{i+1}}{t_i}\right)^{-\alpha}$ for all $i\geq i_0$. Thus, for any sufficiently large $i$, we get
\[ 
	\mathbb{P}(T>t_i) \leq \mathbb{P}(T>t_{i_0}) \prod_{j=i_0}^{i-1} \mathbb{P}(T>t_{j+1} \mid T>t_j) = \mathbb{P}(T>t_{i_0}) t_{i_0}^\alpha t_i^{-\alpha} = O(t_i^{-\alpha}) \ . \qedhere
\]
\end{proof}

\begin{proposition}\label{prop:possibrucondition}
Let $p_0, \alpha>0$ be given. Suppose there exists a constant $a>0$ such that
\[
	\left( 1- 4e^{-\frac{a^2}{4}+2\alpha}\right) \frac{(1-p_0)^2}{2a^2 p_0 \sqrt{3-2p_0}}>\alpha \ .
\]
Then, by setting the parameters $A=a\frac{p_0}{2}\sqrt{3-2p_0}$ and $B=a\frac{p_0}{2}$ in our sweeping strategy, it holds that
$\mathbb{P}(T>t) = O(t^{-\alpha})$.% In particular, putting $a=4.5$, yields $\mathbb{E}[T]<\infty$ for any $p<p_0=1/90$.
\end{proposition}
\begin{proof}
We know by \Cref{prop:chernoffbndgrandma} that the agent will be inside the box $(\pm a \sigma_1 \sqrt{t}, \pm a \sigma_2 \sqrt{t})$ with probability at least $1-4e^{-\frac{a^2}{4}+2\alpha}$. 
Moreover, by \Cref{prop:phasegoodifgrandmanear} we know that, given this event, one phase of our sweeping approach finds home with probability at least $(1-p_0)\frac{\ln t}{\sqrt{t}} - o(1)$, where the iteration takes $\frac{(2+o(1)) a^2 p_0\sqrt{3-2p_0}\sqrt{t}\ln t}{1-p_0}$ time, hence yielding a success probability per time step that, up to lower order terms, equals the left-hand side of the analytic expression in \Cref{prop:possibrucondition}. Under the hypothesis of the proposition, it follows that, eventually, Inequality~\eqref{eq:cap} is saturated, so eventually the success probability of an iteration will always equal $\alpha \frac{t_{i+1}-t_i}{t_i}$. But then, by \Cref{claim:algpowerlaw}, it follows that $\mathbb{P}(T>t)= O(t^{-\alpha})$, as desired.
\end{proof}

\begin{remark*}
It seems likely that the bound $p<0.01139\dots$, as obtained from our analysis, is far from optimal. In fact, we believe that the sweeping strategy described in this section should be able to perform much better than our analysis indicates -- perhaps even close to optimal. The main source of loss seems to be the concentration bounds on $X_t$. However, it seems like new ideas are required to gain a better understanding of the limits of our approach.

%	Most likely, the bound of $1/90$ is far from optimal. However, using our techniques, it is not likely that the numbers can be pushed much further.
	\end{remark*}

% Again, $1/90$ is not optimal, but it cannot be pushed much further without new ideas. Like $1/88.3$.

\section*{Acknowledgements}
We would like the thank Lucas Boczkowski, Sebastian Brandt and Manuela Fischer for fruitful discussions. A special thanks goes to Amos Korman, who introduced the problem.

\bibliographystyle{alpha}
\bibliography{references}
\end{document}